\documentclass{article}

\usepackage{appendix} 
\usepackage{amssymb}
\usepackage{subeqnarray}
\usepackage{bm}
\usepackage{tikz}
\usepackage{dsfont}
\usepackage{upgreek}
\usepackage{booktabs}
\usepackage{fancyhdr}
\usepackage{amsthm}
\usepackage{multirow}
\usepackage{multicol}
\usepackage{float}
\usepackage[inline]{enumitem}

\usepackage{algpseudocode}
\usepackage[ruled,vlined,linesnumbered,inoutnumbered]{algorithm2e}

\usepackage{subfigure}
\usepackage{bbding}
\usepackage{caption}
\usepackage{todonotes}
\usepackage{amsfonts}
\usepackage{mathrsfs}
\usepackage{amsmath}
\usepackage[top=2.5cm, bottom=2.5cm, left=3cm, right=3cm]{geometry} 
\usepackage{graphicx}

\providecommand{\U}[1]{\protect\rule{.1in}{.1in}}

\usepackage[breaklinks,colorlinks,linkcolor=black,citecolor=black,urlcolor=black]{hyperref}

\newtheorem{theorem}{Theorem}[section]

\newtheorem{lemma}[theorem]{Lemma}

\newtheorem{remark}{Remark}

\numberwithin{equation}{section}

\newcommand{\be}{\begin{equation}}
\newcommand{\ee}{\end{equation}}
\newcommand{\bee}{\begin{equation*}}
\newcommand{\eee}{\end{equation*}}
\newcommand{\bea}{\begin{eqnarray}}
\newcommand{\eea}{\end{eqnarray}}
\newcommand{\beaa}{\begin{eqnarray*}}
\newcommand{\eeaa}{\end{eqnarray*}}

\newcommand{\jkh}[1]{\left\langle#1\right\rangle}

\allowdisplaybreaks
\graphicspath{ {./results/} }
\definecolor{Gray}{rgb}{0.5,0.5,0.5}

\makeatletter

\newcommand{\Rmnum}[1]{\expandafter\@slowromancap\romannumeral #1@}
\makeatother

\begin{document}

\title{The Variable Volatility Elasticity Model from Commodity Markets\footnotemark[1]}


\author{Fuzhou Gong\footnotemark[2] $^,$\footnotemark[3]
	\and
	Ting Wang\footnotemark[2] $^,$\footnotemark[3]
}

\renewcommand{\thefootnote}{\fnsymbol{footnote}}
\footnotetext[1]{The research was supported in part by the National Natural Science Foundation of China (No. 11688101).}
\footnotetext[2]{Academy of Mathematics and Systems Science, Chinese Academy of Sciences, Beijing, China.}
\footnotetext[3]{School of Mathematical Sciences, University of Chinese Academy of Sciences, Beijing, China.}

\date{} 
\maketitle

\begin{abstract}
	
	In this paper, we propose and study a novel continuous-time model, 
	based on the well-known constant elasticity of variance (CEV) model,
	to describe the asset price process.
	The basic idea is that the volatility elasticity of the CEV model can not be treated as a constant 
	from the perspective of stochastic analysis.
	To address this issue, we deduce the price process of assets 
	from the perspective of volatility elasticity, 
	propose the constant volatility elasticity (CVE) model,
	and further derive a more general variable volatility elasticity (VVE) model. 
	Moreover, our model can describe the positive correlation between volatility and asset prices
	existing in the commodity markets,
	while CEV model can only describe the negative correlation.
	Through the empirical research on the financial market, 
	many assets, especially commodities, 
	often show this positive correlation phenomenon in some time periods, 
	which shows that our model has strong practical application value. 
	Finally, we provide the explicit pricing formula of European options 
	based on our model. 
	This formula has an elegant form convenient to calculate,
	which is similarly to the renowned Black-Scholes formula
	and of great significance to the research of derivatives market.

\end{abstract}

\section{Introduction}

\label{sec:introduction}

	The continuous-time model is a fundamental and powerful tool 
	in the theoretical research of financial markets,
	where stochastic techniques play a prominent part.
	For example, in option pricing, hedging and portfolio research,
	stochastic analysis techniques are usually employed to derive explicit solutions or closed-form expressions,
	which have a strong guiding significance for the practical applications.

	As early as 1965, Samuelson proposed to 
	describe the change process of stock price by geometric Brownian motion \cite{samuelson}, 
	which established the classic continuous time financial model so far.
	In 1973, Black and Sholes used the Ito formula in stochastic analysis 
	to deduce an explicit European option pricing formula based on this model, 
	that is, the classical Black-Scholes formula \cite{Black1973}. 
	This work made a breakthrough in the development of option pricing theory. 
	However, empirical data shows that ``volatility smile" or skew often appears in most of stock markets \cite{Hagan2002}. 
	Therefore, a large number of researches are devoted to proposing more complex continuous time financial models to correct the volatility bias, 
	such as constant elasticity of variance (CEV) model \cite{Cox1976}, stochastic volatility model \cite{Hull1987}, 
	V.G. model \cite{Madan1990}, GARCH model \cite{Bollerslev1986}, and so on. 

	Recently, 
	the classical CEV model has attracted a lot of attentions 
	due to its wide applications in financial markets. 
	In this model, the volatility is negatively correlated with the price.
	Thus, it can describe the derivatives and volatility of stocks 
	with this reverse phenomenon \cite{beckers1980constant}.
	In addition, many extensions based on the CEV model are proposed,
	such as \cite{carr2006jump,Emanuel1982,heath2002consistent}.
	However, we find that there is a problem with the argument 
	of constant elasticity in the CEV model.
	Before discussing related issues in details,
	we first introduce the concepts of volatility and elasticity.  
	 
%
	 
\subsection{Volatility and Elasticity}

	The volatility of financial assets is a key measure for the uncertainty of its return rate, 
	which represents the risk level of financial assets. 
	Mathematically, it is defined as the conditional standard deviation of the return rate, 
	which can be expressed as follows.
	\begin{equation*}
		v_t = \sqrt{ \mathrm{Var} \left(  r_{t+1} \mid \mathcal{F}_t \right) }.
	\end{equation*}
	Here, $ r_{t} $ represents the rate of return at time $t$, 
	and the filtration $ (\mathcal{F}_t) $ covers all the information in the market before time $t$. 
	In this paper, we assume that the market is efficient, that is, the asset price completely reflects all the information in the market.
	The volatility plays an important role in many economic decisions of financial markets,
	such as risk management, derivatives pricing, and portfolio decision-making. 
	In addition, derivatives based on the volatility index are becoming more and more active in the market, 
	which can provide investors with a variety of investment and hedging tools, 
	help to avoid risks and reduce the irrational volatility of the market \cite{Black2006}. 
	In developed financial markets, 
	a variety of volatility indices and their derivatives are widely traded,
	such as futures, options and ETFs based on CBOE Volatility Index (VIX). 
	Therefore, the study of volatility is of great significance in both practice and theory.

	The concept of elasticity, first introduced by Marshall \cite{marshall2013elasticity}, 
	has been widely used in economics.
	It is an important measure of how sensitive an economic factor is to another, 
	for example, changes in supply or demand to the change in price, 
	or changes in demand to changes in income.
	Moreover, the elasticity can help us monitor the risk exposed in the stock markets.
	In this paper, we focus on the elasticity of the volatility 
	for an asset with respect to its price,
	which is called {\it volatility elasticity} and denoted by $\lambda_t$.

	According to the definition of elasticity, 
	$\lambda_t$ measures the sensitivity of volatility to a change in the asset price,
	which can be expressed as follows.
	\begin{equation*}
		\lambda_t =  \left. \frac{ \bigtriangleup v_t }{v_t} \middle/ \frac{ \bigtriangleup S_t }{S_t}. \right.
	\end{equation*}
	In essence, the asset price and its volatility vary according to the market conditions.
	More specifically, with the passage of time, 
	various event information in the market continues to occur, 
	making asset prices fluctuate constantly.
	Therefore, the common driving force of these two stochastic processes is the time scale $t$.
	As $ \bigtriangleup t \to 0$, the continuous formula of volatility elasticity becomes the following form.
	\begin{equation}
		\lambda_t = \left. \frac{ d v_t }{v_t} \middle/ \frac{ d S_t }{S_t}, \right.
	\end{equation}
	where $ d v_t $ and $ d S_t $ denote the stochastic differential of $ v_t $ and $ S_t $, respectively.
	
\subsection{Existing Problems in the CEV Model}

	Let $S_t$ denote the price of the risk asset.
	In the CEV model \cite{Cox1976}, 
	the price process satisfies the following stochastic differential equation (SDE).
	\begin{equation}
		dS_t = S_t \left[\mu dt + \sigma S_t^{\beta / 2 - 1} dB_t \right],
	\end{equation}
	where $\mu$ and $\sigma$ are two constants,
	$\beta \in (0,2]$ is called the elastic factor,
	and $B_t$ represents the Brown motion.
	It has been shown in \cite{Cox1976} that
	the volatility of the asset price is $v_t = \sigma S_t^{\beta / 2 - 1}$
	and the volatility elasticity is equal to a constant $\beta / 2 - 1$.
	We note that these results are derived in the deterministic sense, 
	that is, $v_t$ and $S_t$ are assumed to be deterministic variables and the derivative of $v_t$ with respect to $S_t$ is computed straightforwardly.		
 	However, according to the definition of elasticity in economics,
    $v_t$ and $S_t$ here are random processes intrinsically in the real markets.
    We should use stochastic differential to deal with the volatility elasticity,
    which can not be simply regarded as a constant.
	Consequently, the main goal of this paper is to devise a new model 
	based on the stochastic analysis to describe the real markets.
	
\subsection{Contribution}
	
	In this paper, we design two SDE models to describe the asset price 
	by employing  stochastic differential tools.
	The first model is based on the assumption that the volatility elasticity is constant,
	which is called constant volatility elasticity (CVE) model.
	And we deduce the second model, 
	which is called variable volatility elasticity (VVE) model,
	by assuming that the elasticity is related to price.
	Through the analysis of the actual market data, 
	our models describe the price process of some special assets well,
	especially the commodities, 
	which provide great help for forecasting 
	or guiding economic decision-making tasks in these markets. 
	Finally, we provide an explicit option pricing formula based on the proposed models, 
	which is of great value to the practical application of the derivatives markets.
	
\subsection{Organization}

	The rest of this paper is organized as follows.
	In Section~\ref{sec:CVE}, we deduce the stochastic differential equation model 
	of asset price process for the case that the volatility elasticity is constant.
	Then this model is extended to the generic case 
	when the volatility elasticity is time-varying in Section~\ref{sec:VVE}.
	The potential of our proposed models is verified from 
	empirical research of actual market data in Section~\ref{sec:analysis}.
	In addition, an explicit pricing formula of European options 
	is presented in Section~\ref{sec:pricing}.
	Finally, we give the conclusion of this paper in the last section.

\section{Constant Volatility Elasticity (CVE) Model}
  
\label{sec:CVE}
	
	In this paper, we consider the scenario that 
	the stock market is composed of two kinds of assets: risk and risk-free.
	Without loss of generality, we assume that 
	the price process of assets in the market is a time-homogeneous Markov process 
	and the assets do not pay dividends. 
	Suppose the price process of the risk-free asset satisfies
	\begin{equation}
	\label{eq:riskfree}
	d\beta_t = r \beta_t dt,
	\end{equation}
	where $ r $ is the risk-free rate.

	In the financial market, the stock volatility is time-varying 
	and often presents volatility clustering 
	and thick tail of logarithmic price return distribution \cite{Cont2007volatility}. 
	And according to the efficient market hypothesis and assumption of time-homogeneous Markov property of price process,
	the volatility can be expressed as a function of price
	by its definition.
	Let $S_t$ denote the price of the risk asset.
	Then we assume that $S_t$ satisfies the following stochastic process.
	\begin{equation}
		\label{eq:assume}
		dS_t = S_t \left[\mu dt + \theta (S_t) dB_t \right],
	\end{equation}
	where $ S_0 > 0$ and $ \mu > 0 $ is the instantaneous expected rate of return.
	We assume that $\theta (S_t)$ satisfies some appropriate conditions 
	to guarantee the existence of solutions. 
	It can be readily verified that $ \theta (S_t) $ is exactly the volatility function of $S_t$.
	Moreover, the instantaneous rate of return of risk assets is simply assumed to be constant
	for the sake of convenience.
	In fact, our conclusion can easily be extended to general cases,
	which will be discussed in Remark~\ref{remark}.
	

	Now, based on the assumption that the volatility elasticity is constant,
	it can be derived that $\theta (S_t)$ admits a closed-form solution.

	\begin{theorem}
		\label{thm:constant}
		Suppose the price of a risk asset satisfies the SDE \eqref{eq:assume}, 
		$\theta(x)$ is twice continuously differentiable,
		and the volatility elasticity is constant.
		Then it holds that
		\begin{equation}
			\theta (S_t) = C S_t^{\alpha},
		\end{equation}
		where  $\alpha \in \{0, 1\}$ is the value of volatility elasticity, and $C$ is a positive constant.
	\end{theorem}

	\begin{proof}
		Since the price satisfies the SDE \eqref{eq:assume},
		we know that $\theta(S_t)$ is the volatility of $S_t$.
		Then we apply the It$\hat{ \text{o} }$'s rule to $\theta(S_t)$ and obtain that
		\begin{align*}
			dv_t & = d\theta (S_t) \\
			& = \theta^{\prime} (S_t) dS_t + \frac{1}{2}\theta^{\prime \prime} (S_t) d\jkh{ S }_t \\
			& = \theta^{\prime} (S_t) S_t \left[\mu dt + \theta (S_t) dB_t \right] 
			+ \frac{1}{2}\theta^{\prime \prime} (S_t) [S_t \theta (S_t)]^2 dt \\
			& = \theta (S_t) \left[  \left( \frac{ \theta^{\prime} (S_t) }{\theta(S_t)}  S_t\mu 
			+ \frac{1}{2}\theta^{\prime \prime} (S_t) S_t^2 \theta (S_t) \right) dt + \theta^{\prime} (S_t) S_t dB_t \right],
		\end{align*}
		which yields that
		\begin{equation}
			\frac{ dv_t }{v_t} = \left( \frac{ \theta^{\prime} (S_t) }{\theta(S_t)}  S_t\mu 
			+ \frac{1}{2}\theta^{\prime \prime} (S_t) S_t^2 \theta (S_t) \right) dt + \theta^{\prime} (S_t) S_t dB_t.
		\end{equation}
		We assume that the volatility elasticity is a constant $\alpha$.
		According to the definition of elasticity, it follows that 
		\begin{equation}
			\frac{ dv_t }{v_t} = \alpha\frac{ dS_t }{S_t},
		\end{equation}
		which further implies that
		\begin{equation}
			\left( \frac{ \theta^{\prime} (S_t) }{ \theta(S_t) }  S_t\mu 
			+ \frac{1}{2}\theta^{\prime \prime} (S_t) S_t^2 \theta (S_t) \right) dt + \theta^{\prime} (S_t) S_t dB_t 
			= \alpha \left[\mu dt + \theta (S_t) dB_t \right] ,
		\end{equation}
		Therefore, we can deduce that 
		\begin{equation}
		\label{eq:solvetheta1}
			\theta^{\prime} (S_t) S_t = \alpha \theta (S_t),
		\end{equation}
		and
		\begin{equation}
		\label{eq:solvetheta2}
			\frac{ \theta^{\prime} (S_t) }{ \theta(S_t) }  S_t\mu 
			+ \frac{1}{2}\theta^{\prime \prime} (S_t) S_t^2 \theta (S_t) 
			= \alpha \mu.
		\end{equation}
		Clearly, \eqref{eq:solvetheta1} is an ordinary differential equation (ODE),
		and its solution is 
		\begin{equation}
		\label{eq:theta1}
			\theta(S_t) = C S_t^\alpha,
		\end{equation}
		where $C$ is a positive constant.
		Substituting \eqref{eq:theta1} to \eqref{eq:solvetheta2}, we have
		\begin{equation*}
	    	\frac{1}{2} \alpha (\alpha - 1) C^2 S_t^{2 \alpha} = 0.
		\end{equation*}
		Since $S_t$ is not identically equal to $0$, we can conclude that
		\begin{equation}
			\alpha = 0, \text{~or~} \alpha = 1, \text{~or~} C = 0.
		\end{equation}
		If $C = 0$, the SDE \eqref{eq:assume} becomes $dS_t = S_t \mu dt$. 
		This means that $S_t$ is the price of a risk-free asset, 
		which contradicts the assumption.
		Therefore, we can acquire that
		\begin{equation}
			\theta (S_t) = C S_t^{\alpha},
		\end{equation}
		where $\alpha \in \{0, 1\}$, and $C$ is a positive constant.
	\end{proof}

	Theorem 3.1 shows that when the volatility elasticity is constant, 
	this value can only be 0 or 1. 
	And the price of the risk asset can be expressed as
	\begin{equation}
	\label{eq:bsmodel}
		dS_t = S_t \left[\mu dt + C dB_t \right],
	\end{equation}
	or 
	\begin{equation}
	\label{eq:squaremodel}
		dS_t = S_t \left[\mu dt + C S_t dB_t \right].
	\end{equation}
	
	The first model \eqref{eq:bsmodel} is nothing but the 
	classical Black-Scholes model.
	In this case, the volatility of asset is a constant $\sigma$ and 
	the volatility elasticity is obviously 0. 
	We will not discuss more details about this trivial case in this paper.

	The second model \eqref{eq:squaremodel} is similar to the CEV model in form, 
	which corresponds to the case that $ \beta = 4$. 
	However, it should be noted that $\beta$ lies in $(0, 2]$ in the CEV model \cite{cox1996constant}.
	Therefore, our model is not a special case of it.
	Since the volatility elasticity is assumed to be constant, 
	\eqref{eq:squaremodel} is called Constant Volatility Elasticity (CVE) model.
	
	\begin{remark}
	\label{remark}
		In the proof of Theorem \ref{thm:constant},
		the constant $\mu$ can be replaced by an adaptive process $\mu(t)$ 
		and the conclusion will not change.
	\end{remark}

\section{Variable Volatility Elasticity (VVE) Model}

\label{sec:VVE}
	
	In the financial market, the return series of risk asset prices 
	usually have the phenomenon of
	excessive volatility and volatility clustering \cite{Cont2007volatility}.
	A large number of models tries to capture this phenomenon,
	such as stochastic volatility model \cite{Hull1987}, 
	GARCH \cite{Bollerslev1986}, and EGARCH \cite{nelson1991conditional}.
	In the real market, the characteristics of volatility series 
	are often complex and elusive. 
	Therefore, it is too trivial to assume that the elasticity of volatility 
	with respect to price is a constant.
	
	In this section, we will discuss the stochastic differential equation model 
	when the volatility elasticity changes with market information. 
	Recall that the market is assumed to be efficient,  
	that is, the market information flow is generated by the asset price.
	Therefore, we naturally assume that the volatility elasticity depends on 
	the asset price.
	Then, the following theorem derives the resulting model.

	\begin{theorem}
		\label{thm:nonconstant}
		Suppose the price of a risk asset satisfies the SDE \eqref{eq:assume}, 
		$\theta(x)$ is twice continuously differentiable,
		and the volatility elasticity depends on the price. 
		Then it holds that
		\begin{equation}
		\label{eq:mainmodel}
		dS_t = S_t \left[\mu dt +  \left(\sigma + C_1 S_t\right) dB_t \right].
		\end{equation}
		where $\sigma$ and $C_1$ are two positive constants.
	\end{theorem}

	\begin{proof}
		According to the proof of Theorem \ref{thm:constant}, 
		we can obtain that the volatility of $S_t$ is $ \theta(S_t) $ and 
		\begin{equation}
			\frac{ dv_t }{v_t} = \left( \frac{ \theta^{\prime} (S_t) }{ \theta(S_t) }  S_t\mu 
			+ \frac{1}{2}\theta^{\prime \prime} (S_t) S_t^2 \theta (S_t) \right) dt + \theta^{\prime} (S_t) S_t dB_t.
		\end{equation}
		Suppose the volatility elasticity is $\alpha (S_t)$, namely,
		\begin{equation}
			\frac{ dv_t }{v_t}= \alpha(S_t) \frac{ dS_t }{S_t},
		\end{equation}
		which infers that
		\begin{equation}
			\left( \frac{ \theta^{\prime} (S_t) }{ \theta(S_t) }  S_t\mu 
			+ \frac{1}{2}\theta^{\prime \prime} (S_t) S_t^2 \theta (S_t) \right) dt + \theta^{\prime} (S_t) S_t dB_t 
			= \alpha (S_t) \left[\mu dt + \theta (S_t) dB_t \right].
		\end{equation}
		Therefore, we can deduce that 
		\begin{equation}
		\label{eq:solvetheta21}
			\theta' (S_t) S_t = \alpha(S_t) \theta (S_t),
		\end{equation}
		and
		\begin{equation}
		\label{eq:solvetheta22}
			\frac{ \theta'(S_t) }{ \theta(S_t) }  S_t\mu + \frac{1}{2}\theta''(S_t) S_t^2 \theta (S_t) =  \alpha(S_t) \mu.
		\end{equation}
		By solving the ODE \eqref{eq:solvetheta21}, we can obtain that
		\begin{equation}
		\label{eq:theta21}
			\theta(S_t) = C_2 \exp \left\{ \int_{0}^{S_t} \frac{\alpha (z)}{z} dz \right\},
		\end{equation}
		where $C_2$ is a positive constant. 
		Now it is straightforward to verify that
		\begin{equation}
		\label{eq:dtheta}
			\theta^{\prime} (S_t) = C_2 \exp \left\{ \int_0^{S_t} \frac{\alpha (z)}{z} dz \right\} \cdot \frac{\alpha (S_t)}{S_t},
		\end{equation}
		and
		\begin{equation}
		\label{eq:ddtheta}
			\theta^{\prime \prime} (S_t) = C_2 \exp \left\{ \int_0^{S_t} \frac{\alpha (z)}{z} dz \right\} 
			\cdot \frac{\alpha^2 (S_t) + \alpha^{\prime} (S_t) S_t - \alpha(S_t)}{S_t^2}.
		\end{equation}
		Substituting the above three relationships into \eqref{eq:solvetheta22} yields that
		\begin{equation*}
			\frac{1}{2} \theta^2 (S_t) \left[ \alpha^2 (S_t) + \alpha^{\prime} (S_t) S_t - \alpha (S_t) \right] = 0.
		\end{equation*}
		which combining the fact that $\theta (S_t) > 0$ implies that
		\begin{equation*}
			\alpha^2 (S_t) + \alpha^{\prime} (S_t) S_t - \alpha (S_t) = 0.
		\end{equation*}
		By solving the above ODE, we can obtain that
		\begin{equation}
			\alpha (S_t) = \frac{C_3 S_t}{1 + C_3 S_t}, 
		\end{equation}
		where $C_3$ is a positive constant.
		Therefore, we can further conclude that
		\begin{align*}
			\theta (S_t) & = C_2 \exp \left\{ \int_0^{S_t} \frac{\alpha (z)}{z} dz \right\}\\
			& = C_2 ( 1 + C_3 S_t).
		\end{align*}
		The proof is completed by setting $C_1 = C_2 C_3$ and $\sigma = C_2$.
	\end{proof}

	Theorem \ref{thm:nonconstant} establishes the model for the case that 
	the volatility elasticity is not a constant and depends on the price.
	Therefore, \eqref{eq:mainmodel} is called Variable Volatility Elasticity (VVE) model.
	It is obvious that \eqref{eq:squaremodel} is a special case of \eqref{eq:mainmodel} 
	for $\sigma = 0$.
	Therefore, in the sequel, we will focus on the VVE model.
	
\subsection{Existence of Solutions}

	Now, we verify the existence of solutions to the SDE \eqref{eq:mainmodel} 
	in the following lemma, 
	which proves that the proposed model is well-defined.

\begin{lemma}
	\label{lem:solution}
	There exists a solution to the SDE \eqref{eq:mainmodel}, which can be expressed as: 
	\begin{equation}
	S_t = \frac{ \sigma S_0 e^{ \gamma t +\sigma B_t}  }{ \left( \frac{ \mu }{ \gamma } -1\right) C_1 S_0 e^{ \gamma t + \sigma B_t } -  \frac{ \mu }{ \gamma } C_1 S_0 e^{\sigma B_t} + \sigma + C_1 S_0 },
	\end{equation}
	where $\gamma = \mu - \frac{\sigma^2}{2} $.
\end{lemma}
\begin{proof}
	We use the nonlinear transformation method to solve the SDE \eqref{eq:mainmodel}.
	Suppose $Y_t = F(S_t)$ and $F (x)$ is a twice continuously differentiable function. 
	Then it can be readily checked that
	\begin{align*}
	dY_t & = F^{\prime} (S_t) dS_t + \frac{1}{2} F^{\prime \prime} (S_t) d\jkh{S}_t \\
	& = F^{\prime} (S_t) \left[ \mu S_t dt + S_t (\sigma + C_1S_t) dB_t \right] 
	+ \frac{1}{2} F^{\prime \prime} (S_t) S_t^2 (\sigma + C_1 S_t)^2 dt \\
	& = \left[ \mu F^{\prime} (S_t) S_t + \frac{1}{2} F^{\prime \prime} (S_t) S_t^2 (\sigma + C_1 S_t)^2 \right] dt 
	+ F^{\prime} (S_t) S_t (\sigma + C_1S_t) dB_t,
	\end{align*}
	Upon taking 
	\begin{equation}
	\label{eq:solveF}
		F^{\prime} (x) x (\sigma + C_1 x) = \sigma F(x),
	\end{equation}
    we can obtain that
	\begin{equation*}
		F^{\prime} (S_t) S_t (\sigma + C_1S_t) dB_t = \sigma F(S_t) dB_t = \sigma Y_t dB_t.
	\end{equation*}
	And the solution to the ODE \eqref{eq:solveF} can be represented as:
	\begin{equation}
	\label{eq:F}
	F(x) = \frac{Ax}{\sigma + C_1 x},
	\end{equation}
	where $A$ is a constant. 
	Without loss of generality, we fix $A = 1$. 
	Then we have 
	\begin{equation*}
		Y_t = \frac{S_t}{\sigma + C_1 S_t},
	\end{equation*}
	and
	\begin{equation}
	\label{SandY}
		S_t = \frac{\sigma Y_t}{1-C_1 Y_t}.
	\end{equation}
	And it is straightforward to verify that
	\begin{equation*}
		F^{\prime} (x) = \frac{\sigma}{x(\sigma + C_1 x)} F(x),
	\end{equation*}
	and
	\begin{align*}
		F^{\prime \prime} (x) & = \frac{\sigma}{x(\sigma + C_1 x)}F^{\prime}(x) 
		- \frac{\sigma (\sigma + 2 C_1 x)}{x^2 (\sigma + C_1 x)^2} F(x) \\
		& = \frac{\sigma^2}{x^2(\sigma + C_1 x)^2} F(x) - \frac{\sigma^2 + 2 C_1 \sigma x}{x^2 (\sigma + C_1 x)^2} F(x)\\
		& = \frac{-2 C_1\sigma x}{x^2 (\sigma + C_1 x)^2} F(x),
	\end{align*}
	which further implies that
	\begin{equation*}
		\mu F^{\prime} (S_t) S_t = \mu \frac{\sigma}{(\sigma + C_1 S_t)} F(S_t)
		= \frac{\mu \sigma}{(\sigma + C_1 S_t)} Y_t,
	\end{equation*}
	and
	\begin{equation*}
		\frac{1}{2} F^{\prime \prime} (S_t) S_t^2 (\sigma + C_1 S_t)^2 = -F(S_t) C_1\sigma S_t 
		= - \sigma C_1 S_t Y_t.
	\end{equation*}
	Thus, we can obtain that
	\begin{align*}
		dY_t & = Y_t \left[\left( \frac{\mu \sigma}{\sigma + C_1 S_t} - \sigma C_1 S_t \right) dt + \sigma dB_t \right]\\
		&= Y_t\left[\left( \frac{\mu \sigma}{\sigma + C_1  \sigma Y_t/ (1-C_1 Y_t)} - \sigma C_1  \frac{\sigma Y_t}{1-C_1 Y_t} \right) dt + \sigma dB_t \right]\\
		&= Y_t\left[\left( \mu (1-C_1 Y_t) -   \frac{\sigma^2  C_1 Y_t}{1-C_1 Y_t} \right) dt + \sigma dB_t \right]\\
		&= Y_t\left[ \frac{\mu - (2\mu + \sigma^2) C_1 Y_t + \mu C_1^2 Y_t^2 }{1-C_1 Y_t} dt + \sigma dB_t \right].
	\end{align*}
	Next we solve the following SDE.
	\begin{equation}
	\label{eq:YSDE}
		\left\{
		\begin{aligned}
		dY_t &= \frac{\mu - (2\mu + \sigma^2) C_1 Y_t + \mu C_1^2 Y_t^2 }{1-C_1 Y_t} Y_tdt + \sigma Y_tdB_t, \\
		Y_0 &=\frac{S_0}{\sigma + C_1 S_0}.
		\end{aligned}
	\right.
	\end{equation}
	Let $\hat{\sigma}(Y_t) = \sigma Y_t$ and
	\begin{equation*}
		b(Y_t) = \frac{\mu - (2\mu + \sigma^2) C_1 Y_t + \mu C_1^2 Y_t^2 }{1-C_1 Y_t} Y_t.
	\end{equation*}
	We first solve the following ODE.
	\begin{equation}
	\label{eq:YODE}
		\left\{
		\begin{aligned}
		\frac{dy}{d\omega}&=\hat{\sigma}( y(\omega) ) = \sigma y(\omega), \\
		Y_0 &= \xi.
		\end{aligned}
	\right.
	\end{equation}
	And the solution is $y(\omega)=\xi e^{\sigma \omega}$.
	Let $\Phi(\omega, \xi) = y(\omega) = \xi e^{\sigma \omega}$.
	Then we solve the following ODE with the parameter $\omega$.
	\begin{equation*}
	\left\{
		\begin{aligned}
		\frac{d\xi_t}{dt} & = \exp\left\{ - \int_{0}^{B_t(\omega)} \sigma ds \right\} 
		\left( b\left(\Phi(B_t(\omega), \xi_t)\right)
		 - \frac{1}{2} \hat{\sigma}\left( \Phi(B_t(\omega), \xi_t) \right) 
		 \cdot\hat{\sigma}^{\prime}\left( \Phi(B_t(\omega), \xi_t) \right) \right),\\
		\xi_0 & = Y_0.
		\end{aligned}
	\right.
	\end{equation*}
	In fact, the above ODE can be simplified as follows. 
	\begin{equation*}
	\left\{
		\begin{aligned}
		& \frac{2}{2r - \sigma^2} 
		\left( \frac{1}{\xi_t} - \frac{1}{\xi_t - (2r-\sigma^2) / (2r C_1 e^{\sigma B_t(\omega)} ) } \right) d \xi_t = dt,\\
		& \xi_0 = Y_0.
		\end{aligned}
	\right.
	\end{equation*}
	It is straightforward to verify that the solution is
	\begin{equation}
		\xi_t = \frac{Y_0 e^{ \left( \mu - \frac{\sigma^2}{2}\right)t }}{ \frac{2 \mu C_1 Y_0  }{ 2 \mu - \sigma^2} \left( e^{ \left( \mu - \frac{\sigma^2}{2}\right)t } -1 \right)e^{\sigma B_t(\omega)} +1 } .
	\end{equation}
	Thus, we can obtain that
	\begin{equation}
		Y_t(\omega) = \Phi(B_t(\omega), \xi_t(\omega)) = \frac{Y_0 e^{ \left( \mu - \frac{\sigma^2}{2}\right)t +\sigma B_t(\ omega)} }{ \frac{ \mu C_1 Y_0  }{ \mu - \frac{\sigma^2}{2} } \left( e^{ \left( \mu - \frac{\sigma^2}{2}\right)t + \sigma B_t(\omega) } - e^{\sigma B_t(\omega)} \right)  +1 }
	\end{equation}
	is a solution to the SDE \eqref{eq:YSDE}. 
	Finally, it follows from the relationship \eqref{SandY} that
	\begin{equation}
	S_t = \frac{ \sigma S_0 e^{ \gamma t +\sigma B_t}  }{ \left( \frac{ \mu }{ \gamma } -1\right) C_1 S_0 e^{ \gamma t + \sigma B_t } -  \frac{ \mu }{ \gamma } C_1 S_0 e^{\sigma B_t} + \sigma + C_1 S_0 },
	\end{equation}
	where $\gamma = \mu - \frac{\sigma^2}{2} $.
	
	We complete the proof.
\end{proof}

	In Lemma \ref{lem:solution}, we skillfully make nonlinear reductions 
	to the SDE \eqref{eq:theta21}, 
	then directly solve the transformed SDE, 
	and then reversely transform the solution obtained to that of the original model.
	Lemma \ref{lem:solution} not only proves the existence of solutions, 
	but also presents a closed-form expression, 
	which is very helpful for us to price derivatives based on our model.

\subsection{Qualitative Analysis}

	In this subsection, we make a qualitative analysis of our model
	from an economic perspective 
	to show its effectiveness and significance.
	
	{\it Volatility.} 
	The volatility of our VVE model is $v_t=\sigma + C_1 S_t$, where $\sigma$ and $C_1$ are two positive constants. 
	It can describe the stock market with positive correlation between volatility and stock price, 
	which is different from the CEV model. 
	This positive correlation has already been investigated in the literature.
	For instance, Emanuel and Macbeth \cite{Emanuel1982} 
	verified its existence based on the real market data.
	They analyzed the closing price data of Avon products, Eastman Kodak, international business machines and Xerox from 1976 to 1978.
	Through the estimation of parameters in the CEV model, 
	it was concluded that there was a positive correlation between the volatility and price in 1978. 
	Different from financial securities,
	commodity market possesses its intrinsic physical attributes.
	Thus, the price fluctuations of commodities usually show their own distinguishing features.
	More importantly, Geman points out that,
	for the great majority of commodities,
	the volatility is positively related to its price \cite{Geman2009},
	which will be described in detail in the next section. 
	Therefore, our model will be of great use in the pricing and hedging of commodity derivatives.

	{\it Pricing of contingent claims.}
	The derivatives play an important role in the financial field. 
	If measured by the underlying assets, 
	the scale of derivatives market is far larger than the stock market. 
	At present, the total value of the underlying assets of derivatives that have not been closed 
	is many times of the world's total economic output value \cite{Hull2003}.
	Derivative pricing theory is one of the most important research topics in the field of derivatives.
	Under this model, we can derive the explicit pricing formula of European options (given in Section \ref{sec:pricing}). 
	It is worth mentioning that this formula has a similar form to Black-Scholes pricing formula, 
	which can be easily calculated in practice. 
	

\section{Empirical Analysis of Commodity Markets}

\label{sec:analysis}

	In the financial investment market, 
	bulk commodities refer to the commodities that are homogeneous, tradable,
	and widely used as basic industrial raw materials, 
	such as crude oil, non-ferrous metals, steel, agricultural products, iron ore, coal, etc. 
	It includes three categories: agriculturals, metals and energies.
	Commodities constitute the only spot market in human history, 
	which is closely related to the national economy and the people's livelihood. 
	Since most commodities are the basis of industrial production, 
	the changes of futures and spot prices reflecting their supply and demand will directly affect the whole economic system.
	As an important part of the financial market, 
	the research on commodity and their derivatives is of great significance.
	Commodities have both capital and goods attribute, 
	which leads to the great difference between commodities and securities 
	in the price characteristics.
	
	In the commodity market, 
	the theory of storage is often used to explain the spot price volatility, 
	since the supply-demand relationship caused by the physical properties of commodities 
	is the main reason affecting their price changes \cite{Geman2009}.
	
	The most important result of storage theory is that 
	there is a negative correlation between the inventory level and the commodity volatility. 
	Furthermore, according to the theory of supply-demand, 
	the commodity price is negatively correlated with inventory. 
	Therefore, there is a positive correlation between the commodity price and the volatility, 
	which is distinguished from the negative correlation commonly existing in the stock market.
	
	The production capacity and inventory level are two key factors 
	in predicting the commodity prices.
	Fama and French \cite{Fama1987} make statistical analysis on the data of 21 kinds of commodities
	(including wood, animal, metals and agricultural products), 
	and conclude that the variance of commodity prices decreases 
	with the increase of inventory level.
	Geman and Nguyen \cite{Geman2005} use the soybean data of 
	the United States, Brazil and Argentina 
	to reconstruct a monthly, quarterly and annual global soybean database, 
	and show that the commodity price volatility is an increasing linear function 
	of the inverse inventory.
	Geman \cite{Geman2009} points out that it has the same nature in the energy market. 
	When the estimated oil reserves in the United States or other regions decline, 
	the volatility of oil price will increase sharply, 
	and the price will also rise enormously.
	Deaton and Laroque \cite{Deaton1992} analyze and simulate the annual data 
	of 13 commodities, 
	and find that the conditional variance of price is a non-decreasing function of the price.
	To sum up, a series of literatures show that 
	there is a positive correlation between the price and the volatility 
	in commodity markets.
	
	\begin{figure}[htbp!]
		\centering
		\subfigure[Soybean Meal (SM)]{
			\includegraphics[width=0.4\linewidth]{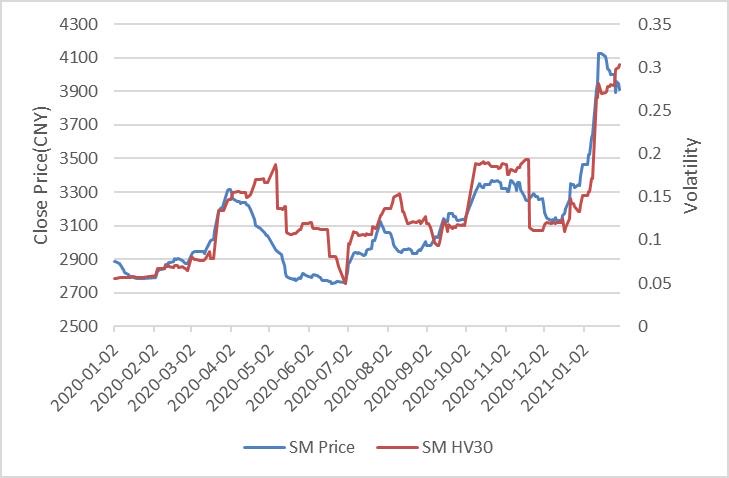}
		}
		\subfigure[Aluminum (AL)]{
			\includegraphics[width=0.4\linewidth]{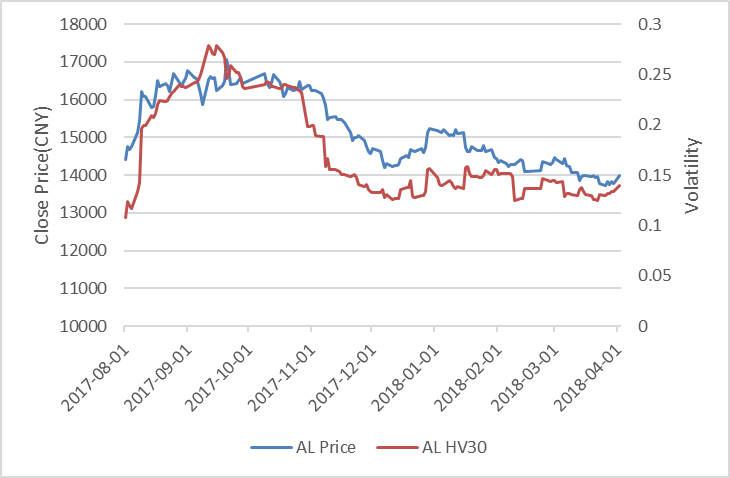}
		}
		\caption{Close price and volatility of soybean meal and aluminum.}
		\label{fig:volatility}
	\end{figure}

	
	\begin{table}[ht!]
		\centering
		\caption{Linear regression results of closing price and volatility data 
			for soybean meal and aluminum.}
		\label{tb:regression} 
		\begin{tabular}{c|c|c} 
			\toprule 
			Commodity & Soybean Meal (SM)  & Aluminum (AL)  \\
			\midrule
			Slope  &  1.53e-04  &  4.61e-05  \\
			\midrule
			P-value of Slope  &  $<$2e-16  &  $<$2e-16   \\
			\midrule
			Intercept  &  -0.346  &  -0.528   \\
			\midrule
			P-value of Intercept  &  $<$2e-16  &  $<$2e-16  \\
			\midrule
			R-squared  &  0.7471  & 0.8207  \\
			\midrule 
			Correlation Coefficient  &  0.8643  &  0.9059  \\
			\bottomrule 
		\end{tabular}
	\end{table}

	Next, we perform an empirical study on China's commodity markets,
	which demonstrates that the linear positive correlation 
	between the volatility and the commodity price does exist in many time intervals.
	Specifically, we select the market data of soybean meal and aluminum for analysis. 
	For soybean meal, we select the closing price data of 263 trading days
	(2020/01/02-2021/01/29).
	And for aluminum, we select the closing price data of 164 trading days
	(2017/08/01-2018/04/02). 
	All the data are collected from wind.
	
	In Figure~\ref{fig:volatility}, the daily frequency closing price data and historical volatility 
	(calculated by using the 30 day price data and recorded as HV30) 
	of the two commodities are depicted respectively. 
	Obviously, the trends of two lines show a positive correlation. 
	Furthermore, we analyze the correlation between closing price and volatility data, 
	and conduct linear fitting using R software. 
	The results are shown in Table~\ref{tb:regression}.
	We have the following observations. 
	Firstly, the correlation coefficients of thees two groups of data are very close to 1.
	Secondly, the values of R-squared for these two linear fittings 
	is 0.7471 and 0.8207, respectively.
	Thirdly, the P-values of slope and intercept are all less than 0.01. 
	Thus, in this time interval, the price and the volatility of these two commodities show a linear positive correlation, 
	which is consistent with the VVE model. 
	We can conclude that the VVE model can be used as 
	an approximate continuous-time model of commodity prices, 
	which is helpful to the study of commodities and commodity derivatives markets.


\section{Option Pricing}

\label{sec:pricing}

	In this section, we consider the option pricing under the VVE model \eqref{eq:mainmodel}. 
	The following theorem provides an explicit pricing formula of European call options.

\begin{theorem}
	\label{thm:pricing}
	Suppose the market is complete and there is a unique equivalent martingale measure. 
	Let $\xi = (S_T - K)^+$ be a reproducible European contingent claim, where $K$ is the exercise price. 
	Then its price process is $V_t = C(t,S_t)$, where
	\begin{align*}
		C(t,x)  =& \sigma S_0 e^{-r(T-t)} E[g(Z)\mathds{1}_{(d,+\infty)}] -K e^{-r(T-t)}(1- N(d)).                   
	\end{align*}
	where $\delta = \frac{r}{r-\frac{\sigma^2}{2} } $, 
	$d=\frac{f_T^{-1}(K)-f_t^{-1}(x)}{\sqrt{T-t}}$, 
	$g(z)=\frac{\sigma S_0  e^{-r(T-t)} }{(\delta -1- \delta e^{-r \delta T}) C_1 S_0+ (\sigma + C_1 S_0) e^{-\sigma (z\sqrt{T-t}+f_t^{-1}(x) )} }$, $f_t^{-1}(x)=\frac{1}{\sigma } \ln{ \left[\left( \left( \frac{ r }{ r - \sigma^2 / 2 } -1\right) e^{ (r - \sigma^2 / 2) t} - \frac{ r }{ r - \sigma^2 / 2  }  \right) \frac{C_1 x }{\sigma} \right] } - \frac{1}{\sigma } ( r - \frac{\sigma^2}{2} ) t $. 
	In addition,  $Z$ obeys standard normal distribution, and $N(d)$ denotes standard normal distribution function.
\end{theorem}

\begin{proof}
	Since we assume that the market is complete, the discounted price process $(e^{-rt}S_t)$ is a martingale under the risk neutral probability measure $P^*$.
	In fact, $P^*$ can be defined as follows.
	\begin{equation}
	\frac{dP^*}{dP}|_{\mathcal{F}_T} = \exp \left\{ -\int_{0}^{t} \frac{ \mu - r }{\sigma + C_1 S_u}du - \frac{1}{2} \int_{0}^{t} (\frac{ \mu - r }{\sigma + C_1 S_u})^2 du \right\}.
	\end{equation}
	According to Girsanov theorem, we can get that
	\begin{equation}  
	 B_t^* = B_t + \int_{0}^{t} \frac{ \mu - r }{\sigma + C_1 S_u} du 
	\end{equation}
	is a Brownian motion under $P^*$, and the SDE \eqref{eq:mainmodel} becomes
	\begin{equation}
	\label{eq:risk neural}
		dS_t = S_t \left[ r dt +  \left(\sigma + C_1 S_t\right) dB_t^* \right].
	\end{equation}
	It follows from Lemma \ref{lem:solution} that the solution of \eqref{eq:risk neural} is 
	\begin{align*}
		S_t &= \frac{ \sigma Y_0 e^{ \left( r - \frac{\sigma^2}{2}\right)t +\sigma B_t^*}  }{ \left( \frac{ r }{ r - \frac{\sigma^2}{2} } -1\right) C_1 Y_0 e^{ \left( r - \frac{\sigma^2}{2}\right)t + \sigma B_t^* } -  \frac{ r }{ r - \frac{\sigma^2}{2} } C_1 Y_0 e^{\sigma B_t^*} +1 }\\
		&=:f_t(B_t^*).
	\end{align*}
	According to the risk neutral pricing formula, 
	we can obtain that the price of reproducible European contingent claim $\xi$ in time $t$ is 
	\begin{equation*}
	V_t = E^* \left[ e^{-r(T-t)} (S_T-K)^+ | \mathcal{F}_t \right].
	\end{equation*}
	By the Markov property of the diffusion process $(S_t)$, we have
	\begin{equation*}
	V_t = E^* \left[ e^{-r(T-t)} (S_T-K)^+ | S_t \right].
	\end{equation*}
	Thus, we can denote $V_t = C(t,S_t)$. Then
	\begin{align*}
	C(t,x) =& E^* \left[ e^{-r(T-t)} (S_T-K)^+ | S_t =x \right]\\
	=& e^{-r(T-t)} E^* \left[ (f_T(B_T^*)-K)^+ | B_t^* = f_t^{-1}(x) \right]\\
	=& e^{-r(T-t)} \int_{-\infty}^{+\infty} (f_T(y)-K)^+ \cdot P_{B^*}(t, f_t^{-1}(x); T,y) dy\\
	=& e^{-r(T-t)} \int_{f_t^-(K)}^{+\infty} (f_T(y)-K) \cdot \frac{1}{ \sqrt{2\pi (T-t)} } e^{-\frac{(y-f_t^{-1}(x))^2}{2(T-t)}} dy\\
	=& e^{-r(T-t)} \int_{f_T^{-1}(K)}^{+\infty} \left( \frac{ \sigma Y_0 e^{ r\delta T +\sigma y}  }{ (\delta -1) C_1 Y_0 e^{ r\delta T + \sigma y } - \delta C_1 Y_0 e^{\sigma y} +1 } -K \right) \cdot \frac{1}{ \sqrt{2\pi (T-t)} } e^{-\frac{(y-f_t^{-1}(x))^2}{2(T-t)}} dy\\
	=& \frac{ e^{-r(T-t)} }{ \sqrt{2\pi (T-t)} }  \int_{f_T^{-1}(K)}^{+\infty}  \frac{ \sigma Y_0 e^{ r \delta T } e^{ \sigma y-\frac{(y-f_t^{-1}(x))^2}{2(T-t)}} }{ ( \delta-1 ) C_1 Y_0 e^{ r \delta T + \sigma y } -  \delta C_1 Y_0 e^{\sigma y} +1 }  dy\\
	& - K e^{-r(T-t)} \int_{f_T^{-1}(K)}^{+\infty} \frac{1}{ \sqrt{2\pi (T-t)} } e^{-\frac{(y-f_t^{-1}(x))^2}{2(T-t)}} dy\\
	=& \frac{\sigma e^{-r(T-t)}}{ C_1 \delta (1-e^{-r\delta T})-C_1 } \int_{d}^{+\infty}  \frac{1}{1+\frac{1}{ C_1 Y_0 \left[ (\delta -1) e^{r\delta T} - \delta \right] }e^{-\sigma (z\sqrt{T-t}+f_t^{-1}(x) )}} \frac{1}{\sqrt{2\pi}} e^{-\frac{z^2}{2}}dz\\
	& -K e^{-r(T-t)} \int_{d}^{+\infty} \frac{1}{\sqrt{2\pi}} e^{-\frac{z^2}{2}}dz\\
	=& \sigma S_0 e^{-r(T-t)} \int_{d}^{+\infty}  \frac{\sigma S_0  e^{-r(T-t)} }{(\delta -1- \delta e^{-r \delta T}) C_1 S_0+ (\sigma + C_1 S_0) e^{-\sigma (z\sqrt{T-t}+f_t^{-1}(x) )} } \frac{1}{\sqrt{2\pi}} e^{-\frac{z^2}{2}}dz\\
	& -K e^{-r(T-t)}(1- N(d))\\
	=& \sigma S_0 e^{-r(T-t)} E[g(Z)\mathds{1}_{(d,+\infty)}] -K e^{-r(T-t)}(1- N(d)),
	\end{align*}
	where $\delta = \frac{r}{r-\frac{\sigma^2}{2} } $, 
	$d=\frac{f_T^{-1}(K)-f_t^{-1}(x)}{\sqrt{T-t}}$, $g(z)=\frac{\sigma S_0  e^{-r(T-t)} }{(\delta -1- \delta e^{-r \delta T}) C_1 S_0+ (\sigma + C_1 S_0) e^{-\sigma (z\sqrt{T-t}+f_t^{-1}(x) )} }$, and
	$f_t^{-1}(x)=\frac{1}{\sigma } \ln{ \left[\left( \left( \frac{ r }{ r - \sigma^2 / 2 } -1\right) e^{ (r - \sigma^2 / 2) t} - \frac{ r }{ r - \sigma^2 / 2  }  \right) \frac{C_1 x }{\sigma} \right] } - \frac{1}{\sigma } (r - \frac{\sigma^2}{2} ) t $. In addition, $N(d)$ denotes standard normal distribution function, and $Z$ obeys standard normal distribution.
\end{proof}

\section{Conclusion}

\label{sec:conclusion}

	In this paper, we find that the volatility elasticity of the CEV model 
	can not be simply treated as a constant.
	To address this issue, 
	we use stochastic analysis tools to deduce the stochastic differential equation model 
	when the volatility elasticity is a constant. 
	Since the volatility elasticity for most of risk assets is not unchanged,  
	we then extend this model to the generic case when the volatility elasticity is time-varying.
	Our model can describe the positive correlation between the volatility and the asset price, 
	which often occurs in the commodity markets.
	In contrast, the CEV model can only describe the negative correlation.
	These theoretical findings are validated by actual market data.
	Furthermore, we derive the explicit pricing formula of European options based on our model, 
	which has the similar form with Black-Scholes formula and is easy to calculate.
	This formula has an important guiding significance in the applications of 
	derivatives pricing in the actual commodity markets.

\bibliographystyle{siam}
\bibliography{VE_Wangting}

\addcontentsline{toc}{section}{References}

\end{document}